\documentclass[a4paper,11pt]{article}
\usepackage[hmargin=3.125cm,vmargin=3cm]{geometry}

\usepackage[bitstream-charter]{mathdesign}
\usepackage{XCharter}
\usepackage[scale=0.92,lf]{FiraMono}

\usepackage[tiny]{titlesec}

\usepackage[utf8]{inputenc}
\usepackage[T1]{fontenc}
\usepackage{microtype}

\usepackage[ruled,noend,noline]{algorithm2e}
\SetKw{KwFrom}{from}

\usepackage{booktabs}
\usepackage[colorlinks]{hyperref}
\usepackage{nicefrac}
\usepackage{enumitem}
\usepackage{tikz}
\usetikzlibrary{arrows.meta,calc,patterns}

\usepackage{relsize}
\newcommand*{\defeq}{\stackrel{\mathsmaller{\mathrm{def}}}{=}}

\newcommand{\bA}{\bar{A}}
\newcommand{\ba}{\bar{a}}
\newcommand{\bi}{i'}

\usepackage{amsmath}
\usepackage{amsthm}
\usepackage[capitalise,noabbrev]{cleveref}
\newtheorem{theorem}{Theorem}
\newtheorem{lemma}[theorem]{Lemma}
\newtheorem{corollary}[theorem]{Corollary}

\renewcommand\footnoterule{\kern-3pt \hrule width 2in height 0.25pt \kern 2.75pt}

\title{Bellman--Ford is optimal for shortest hop-bounded paths}

\author{Tomasz Kociumaka \and Adam Polak}

\date{\small Max Planck Institute for Informatics, Saarland Informatics Campus, Germany \\ \texttt{\{tomasz.kociumaka,apolak\}@mpi-inf.mpg.de}}

\begin{document}

\maketitle

\begin{abstract}
    This paper is about the problem of finding a shortest $s$-$t$ path using at most $h$~edges in edge-weighted graphs. The Bellman--Ford algorithm solves this problem in $O(hm)$ time, where $m$ is the number of edges. We show that this running time is optimal, up to subpolynomial factors, under popular fine-grained complexity assumptions.
    
    More specifically, we show that under the APSP Hypothesis the problem cannot be solved faster already in undirected graphs with non-negative edge weights. This lower bound holds even restricted to graphs of arbitrary density and for arbitrary $h \in O(\sqrt{m})$. Moreover, under a stronger assumption, namely the Min-Plus Convolution Hypothesis, we can eliminate the restriction $h \in O(\sqrt{m})$. In other words, the $O(hm)$ bound is tight for the entire space of parameters $h$, $m$, and $n$, where $n$ is the number of nodes.
    
    Our lower bounds can be contrasted with the recent near-linear time algorithm for the negative-weight Single-Source Shortest Paths problem, which is the textbook application of the Bellman--Ford algorithm.
\end{abstract}

\section{Introduction}

The Bellman--Ford algorithm~\cite{Shimbel55,Ford56,Bellman58} is the textbook solution for the Single-Source Shortest Paths (SSSP) problem in graphs with negative edge weights. It runs in $O(nm)$ time, where $n$ denotes the number of nodes and $m$ is the number of edges. If we limit the outer for-loop (see Algorithm~\ref{alg:bellman-ford}) to only $h \leqslant n - 1$ iterations, the algorithm computes single-source shortest paths that use at most $h$ edges (or \emph{hops}) and runs in $O(hm)$ time.

\bigskip

\begin{algorithm}[H]
\caption{The Bellman--Ford algorithm.}\label{alg:bellman-ford}
\small
$d^{(0)} \gets [+\infty, +\infty, \ldots, +\infty]$\;
$d^{(0)}[s] \gets 0$\;
\For{$i$ \KwFrom $1$ \KwTo $n - 1$}{
  $d^{(i)} \gets d^{(i-1)}$\;
  \ForEach{{\rm edge} $(u, v) \in E$}{
    $d^{(i)}[v] \gets \min \{ d^{(i)}[v], d^{(i-1)}[u] + w(u, v) \}$\;
  }
}
\end{algorithm}

\clearpage

Negative-weight SSSP has seen a lot of improvements over Bellman--Ford's running time: scaling algorithms~\cite{Gabow85,GabowT89,Goldberg95}, which eventually led to $O(n^{\nicefrac{1}{2}}m \log W)$ running time, where $W$ denotes the maximum absolute value of a negative edge weight; interior-point methods for the more general Minimum-Cost Flow problem, which recently led to an almost-linear $O(m^{1+o(1)} \log W)$ time algorithm~\cite{Chen22}; and finally, the recent combinatorial near-linear $O(m \log^8(n) \log W)$ time algorithm~\cite{BernsteinNW22}.

\begin{center}
\it Can we get similar improvements for the problem of finding shortest hop-bounded paths?
\end{center}

This basic question stays embarrassingly open. Even in undirected graphs with only nonnegative edge weights, Bellman--Ford remains the fastest known algorithm for that problem. In this paper, we give a negative answer to the above question, up to subpolynomial factors, under popular fine-grained complexity assumptions.

\subsection{Our results}

Let us begin with formally stating the computational problem that we study: Given a graph $G=(V,E)$ with edge weights $w : E \to \mathbb{Z}$, two distinguished nodes $s, t \in V$, and a non-negative integer $h \in \mathbb{Z}_{\geqslant 0}$, find (the length of) a shortest path from $s$ to $t$ that uses at most\footnote{We could also consider a variant of the problem asking for a \emph{walk} with \emph{exactly} $h$ edges. It is the harder of the two variants (adding a $0$-length self-loop to node $s$ reduces the ``at most $h$'' variant to the ``exactly $h$'' variant), and we prove the hardness of the easier one already.} $h$ edges. We call such paths \emph{$h$-hop-bounded}, or simply \emph{hop-bounded} when $h$ is implicit in the context.

In this paper, we give two fine-grained reductions, each proving that the $O(hm)$ running time of the Bellman--Ford algorithm is conditionally optimal for the problem, up to subpolynomial factors. Our two hardness results differ from each other in (1) how the parameters $n, m, h$ of the hard instances relate to each other, and (2) which hardness assumption is required. Table~\ref{tab:summary} summarizes these differences.

Our first result holds under the APSP Hypothesis.

\begin{theorem}\label{thm:apsp}
Unless the APSP Hypothesis fails, there is neither $O(h^{1-\varepsilon}m)$ nor $O(hm^{1-\varepsilon})$ time algorithm for finding the length of a shortest $h$-hop-bounded $s$-$t$ path in undirected graphs with nonnegative edge weights, for any $\varepsilon > 0$.

This holds even restricted to instances with density $n=\Theta(\sqrt{m})$ and hop bound $h=\Theta(m^\eta)$ for arbitrarily chosen $\eta \in (0, \nicefrac{1}{2}]$.
\end{theorem}

Although the hard instances in \cref{thm:apsp} are dense, one can trivially obtain sparser instances by adding isolated nodes. Indeed, such nodes influence neither the length of a shortest $h$-hop-bounded $s$-$t$ path nor the running time bounds as functions of $h$ and $m$.

\begin{corollary}\label{cor:apsp}
  The result of \cref{thm:apsp} holds even restricted to instances with density $n=\Theta(m^\nu)$ and hop bound $h=\Theta(m^\eta)$ for arbitrarily chosen $\nu \in [\nicefrac{1}{2},1]$ and $\eta \in (0, \nicefrac{1}{2}]$.
\end{corollary}

\begin{table}
\centering
\small
\begin{tabular}{@{}llllll@{}}
\toprule
 & Density & Hops & Hypothesis \\
\midrule
\cref{cor:apsp}   & $n=\Theta(m^{\nu})$, $\nu\in [\nicefrac12,1]$ & $h=\Theta(m^\eta)$, $\eta\in (0,\nicefrac12]$ & APSP \\
\cref{cor:minplus} & $n=\Theta(m^{\nu})$, $\nu\in [\nicefrac12,1]$ & $h=\Theta(m^\eta)$, $\eta\in [\nicefrac12,\nu]$ & Min-Plus Convolution \\
\bottomrule
\end{tabular}
\caption{Summary and comparison of our conditional hardness results for the problem of finding (the length of) a shortest hop-bounded path between two nodes.}\label{tab:summary}
\end{table}

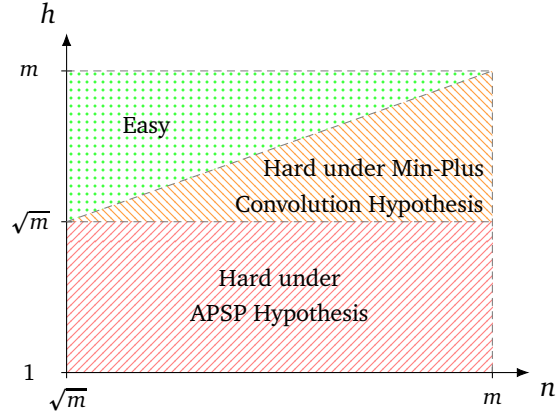
\begin{figure}
\centering
\begin{tikzpicture}[xscale=1.4]

  \fill [pattern color=red!50, pattern=north east lines] (0,0) -- (4,0) -- (4,2) -- (0,2);
  \fill [pattern color=orange!75, pattern=north west lines] (4,4) -- (4,2) -- (0,2);
  \fill [pattern color=green!75, pattern=dots] (4,4) -- (0,4) -- (0,2);

  \draw[-Latex] (-0.05,0) -- (4.35,0) node [below right] {$n$};
  \draw[-Latex] (0,-0.075) -- (0,4.5) node [above left] {$h$};

  \draw[-] (4,-0.075) -- (4,0.075);
  \draw[-] (-0.05,2) -- (0.05,2);
  \draw[-] (-0.05,4) -- (0.05,4);

  \draw[-,densely dashed, black!50] (0,2) -- (4,2);
  \draw[-,densely dashed, black!50] (0,2) -- (4,4);
  \draw[-,densely dashed, black!50] (0,4) -- (4,4) -- (4,0);

  \node [font=\footnotesize] at (0, -0.35) {$\sqrt{m}$};
  \node [font=\footnotesize] at (4, -0.35) {$m$};

  \node [font=\footnotesize] at (-0.35, 0) {$1$};
  \node [font=\footnotesize] at (-0.35, 2) {$\sqrt{m}$};
  \node [font=\footnotesize] at (-0.35, 4) {$m$};

  \node at (0.75,3.25) {\footnotesize Easy};
  \node[align=right,font=\footnotesize] at (2.75,2.45) {Hard under Min-Plus\\ Convolution Hypothesis};
  \node[align=center,font=\footnotesize] at (2,1) {Hard under\\ APSP Hypothesis};

\end{tikzpicture}
\vspace{-1.125em}
\caption{Parameter space. The upper-left triangle represents the case of $h\geqslant n$, where the problem degenerates to the standard Shortest Path problem, without hop bound, which can be solved in $\tilde{O}(m)$ time.}\label{fig:parameters}
\end{figure}

We remark that it is not very surprising that our reduction from APSP can only produce instances with $h \in O(\sqrt{m})$. The conjectured time complexity of APSP is $N^{3 - o(1)} = |\mathrm{input}|^{\nicefrac{3}{2} - o(1)}$. For $h=\Theta(m^\eta)$, the $O(hm)$ time bound is actually $O(|\mathrm{input}|^{1+\eta})$. Fine-grained reductions from problems with smaller complexity to problems with larger complexity are possible (see, e.g.,~\cite{Lincoln20}) but rare, and to our best knowledge no such reduction from APSP is known. If \cref{thm:apsp} worked for $\eta > \nicefrac{1}{2}$, this would be the first such example.

In order to cover the remaining combinations of parameters, we use a stronger hypothesis, concerning a quadratic time problem, namely the Min-Plus Convolution Hypothesis. Since this hypothesis implies the APSP Hypothesis, it is also a sufficient condition for \cref{thm:apsp} and thus gives hardness for the entire parameter space.

\begin{theorem}\label{thm:minplus}
Unless the Min-Plus Convolution Hypothesis fails, there is neither $O(h^{1-\varepsilon}m)$ nor $O(hm^{1-\varepsilon})$ time algorithm for finding the length of a shortest $h$-hop-bounded $s$-$t$ path problem in undirected graphs with nonnegative edge weights, for any $\varepsilon > 0$.

This holds even restricted to instances with density $n=\Theta(m^\eta)$ and hop bound $h=\Theta(m^\eta)$ for arbitrarily chosen $\eta \in [\nicefrac{1}{2}, 1]$.
\end{theorem}

Just like before, we can sparsify the hard instances by adding isolated nodes.

\begin{corollary}\label{cor:minplus}
  The result of \cref{thm:minplus} holds even restricted to instances with density $n=\Theta(m^\nu)$ and hop bound $h=\Theta(m^\eta)$ for arbitrarily chosen $\nu \in [\nicefrac{1}{2},1]$ and $\eta\in [\nicefrac12,\nu]$.
\end{corollary}
Combining \cref{cor:apsp,cor:minplus}, we cover the entire range of parameters $\nu\in [\nicefrac{1}{2},1]$ and $\eta \in (0,\nu]$ for which
the $O(hm)$ running time is optimal; see \cref{fig:parameters}.

Let us point out that, even though the Bellman--Ford algorithm finds paths from a single source $s$ to all the nodes in the graph, the above two hardness results hold even for the easier problem of finding a single path between two distinguished nodes $s$ and $t$. Moreover, note that any (shortest path) algorithm for directed graphs could also be used for undirected graphs (but not necessarily vice versa). Bellman--Ford works in directed graphs with possibly negative edge weights, while our hardness results already hold for undirected graphs with nonnegative edge weights.

\subsection{Hardness assumptions}

In this section, we briefly recall the two hypotheses that we use in our theorems.

The APSP Hypothesis is an assertion that the All-Pairs Shortest Paths~(APSP) problem cannot be solved in truly subcubic $O(n^{3-\varepsilon})$ time for any $\varepsilon > 0$.
It is one of the three main hypotheses in fine-grained complexity, the other two being the 3SUM Hypothesis and Strong Exponential Time Hypothesis (SETH)~\cite{VVW18}. The APSP Hypothesis is strengthened by the existence of a large class of problems that are equivalent to the APSP problem under subcubic reductions~\cite{WilliamsW18,VVW18}, and by the lack of a truly subcubic algorithm for any of these problems.

In the Min-Plus Convolution problem, we are given two sequences $(a[i])_{i=0}^{n-1}$, $(b[i])_{i=0}^{n-1}$, and the goal is to output sequence $(c[i])_{i=0}^{n-1}$ defined as $c[k] = \max_{i+j=k} (a[i] + b[j])$. So far, only subpolynomial $2^{O(\sqrt{\log n})}$-factor improvements~\cite{BremnerCDEHILPT14, Williams18} over the naive quadratic running time are known. The Min-Plus Convolution Hypothesis~\cite{KunnemannPS17,CyganMWW19} states that the problem cannot be solved in truly subquadratic time, that is, $O(n^{2-\varepsilon})$ for any $\varepsilon > 0$.
Similarly to APSP, there is also a class (albeit smaller) of problems equivalent to Min-Plus Convolution under subquadratic reductions~\cite{CyganMWW19}.\footnote{One of the problems equivalent to Min-Plus Convolution is the Knapsack problem on instances with target value $t=\Theta(n)$. Recall that Knapsack can be solved in $O(nt)$ time by a dynamic programming algorithm due to Bellman~\cite{Bellman1956}. Hence, we can half-jokingly rephrase \cref{thm:minplus} and say that if one Bellman's algorithm is optimal for Knapsack, then the other Bellman's algorithm is optimal for shortest hop-bounded paths.}

The two hypotheses are closely related because APSP is runtime-equivalent (up to constant factors) to the Min-Plus Product problem~\cite{FischerM71}, which is the matrix product analogue of Min-Plus Convolution. There is a reduction from the convolution to the product problem~\cite{BremnerCDEHILPT14}, which entails that the Min-Plus Convolution Hypothesis implies the APSP Hypothesis, and the former is therefore a stronger assumption.

As customary in fine-grained complexity, these hypotheses, as well as all the results in this paper, are stated in the word RAM model of computation with $O(\log n)$-bit machine words. We assume all input numbers fit into single machine words. Alternatively, the APSP Hypothesis is sometimes stated as follows~\cite{VVW18}: For every $\varepsilon > 0$ there is a constant $c$ such that APSP cannot be solved in $O(N^{3 - \varepsilon})$ time in $N$-node graphs with edge weights in $\{-N^c,\ldots,N^c\}$. Our results could also be stated this way because our reductions do not increase weights by more than polynomial factors.

\subsection{Related work}

Hop-bounded paths are studied in various areas related to graph algorithms, e.g., distributed algorithms~\cite{GhaffariHZ21}, dynamic algorithms~\cite{Thorup05,LackiN22}, or even polyhedral combinatorics~\cite{DahlG04}. Problems of finding shortest hop-bounded paths appear, e.g., in the context of quality-of-service (QoS) routing in networks~\cite{BalakrishnanA92, ChenN98}.

Guérin and Orda~\cite{GuerinO02} and Cheng and Ansari~\cite{ChengA04} studied the problem of finding shortest $h$-hop-bounded paths from single source $s$ to all nodes in the graph and for all hop bounds $h \leqslant H$. They proved an $\Omega(Hm)$ lower bound for that problem against so-called \emph{path-comparison-based algorithms}, i.e., algorithms that only access the edge weights by comparing the lengths of two paths. 
Although Dijkstra and Bellman--Ford are known to be path-comparison-based algorithms, algebraic algorithms based on fast matrix multiplication are not.

\paragraph{Bicriteria Path.}

In the Bicriteria Path problem, we are given a graph $G=(V,E)$ with two types of non-negative edge weights $l, c : E \to \mathbb{Z}$, called \emph{lengths} and \emph{costs}, respectively; two \emph{budgets} $L, C \in \mathbb{Z}$; and two distinguished nodes $s, t \in V$. The goal is to find a path from $s$ to~$t$ with the total length at most $L$ and the total cost at most~$C$. Joksch's algorithm~\cite{Joksch66} solves the problem in pseudopolynomial $O(\min(L, C) \cdot m)$ time. For the special case of all edge costs equal to $1$, the Bicriteria Path problem is equivalent to the problem we study in this paper and, moreover, Joksch's algorithm runs in the same time as the Bellman--Ford algorithm.

Abboud, Bringmann, Hermelin, and Shabtay~\cite{AbboudBHS22} proved that, unless SETH fails, there is no algorithm solving the Bicriteria Path problem on sparse graphs (with $m=\Theta(n)$ edges) with budgets $L, C = \Theta(n^\gamma)$ in time $O(n^{1+\gamma-\varepsilon})$ for any $\varepsilon > 0$ and $\gamma > 0$. In other words, they proved that Joksch's algorithm is conditionally optimal, up to subpolynomial factors. Their reduction, however, heavily uses both types of edge weights, and hence it does not imply any lower bound for the special case with unit costs, i.e., for the problem of our interest.

\section{Hardness under APSP Hypothesis}

\paragraph{Preliminaries.}
Given a complete tripartite graph $G=(A \cup B \cup C, E)$ with edge weights $w : E \to \mathbb{Z}$, the Negative Triangle problem asks to find three nodes $a \in A$, $b \in B$, and $c \in C$ with $w(a,b) + w(b,c) + w(c,a) < 0$.
Vassilevska Williams and Williams~\cite{WilliamsW18} proved that APSP and Negative Triangle are equivalent under subcubic reductions.
In particular, unless the APSP Hypothesis fails, there is no $O(N^{3-\varepsilon})$ time algorithm for Negative Triangle with $|A|=|B|=|C|=N$, for any $\varepsilon > 0$. 

Via a by now standard argument, under the same assumption, for any $\varepsilon>0$, there is no $O(N^{\alpha + 2 - \varepsilon})$ time algorithm for the problem restricted to instances with $|A|=\Theta(N^\alpha)$ and $|B|=|C|=N$ for arbitrarily chosen $\alpha \in (0, 1]$. 
Specifically, the reduction partitions the original set $A$ into $\Theta(N^{1-\alpha})$ subsets of size $\Theta(N^\alpha)$ each; the sets $B$ and $C$ are copied to all $\Theta(N^{1-\alpha})$ produced instances. 
A negative triangle exists in the original instance if and only if it exists in at least one of the produced instances. 
Thus, if each of the obtained instances could be solved in $O(N^{\alpha+2-\varepsilon})$ time, then the original instance could be solved in $O(N^{1-\alpha}\cdot (N^2 +N^{\alpha+2-\varepsilon}))=O(N^{3-\alpha}+N^{3-\varepsilon})$ time, violating the APSP Hypothesis.

\paragraph{Reduction.}
We show how to reduce an instance of Negative Triangle, with $|A|=\Theta(N^{\alpha})$ and $|B|=|C|=N$, to finding the minimum length of an $h$-hop-bounded $s$-$t$ path in an undirected graph with $n=\Theta(N)$ nodes, $m=\Theta(N^2)$ edges, and the hop bound $h=\Theta(N^{\alpha})$. 
In order to prove \cref{thm:apsp}, we set $\alpha = 2\eta$ so that $n=\Theta(N)=\Theta(\sqrt{m})$ and $h = \Theta(N^{\alpha})=\Theta(N^{2\eta})=\Theta(m^\eta)$.
An $O(h^{1-\varepsilon} m)$ time (or $O(hm^{1-\varepsilon})$ time) algorithm finding a shortest $h$-hop-bounded $s$-$t$ path would thus yield an $O(N^{\alpha(1-\varepsilon)+2})=O(N^{\alpha + 2 - \alpha \varepsilon})$ time (respectively, $O(N^{\alpha + 2 - 2 \varepsilon})$ time) algorithm for the original instance of the Negative Triangle problem. 
As argued above, no such algorithm exists unless the APSP Hypothesis fails. 

\begin{figure}
\centering
\begin{tikzpicture} 
  \foreach \i/\l in {1/1,2/2,7/{P-1},8/P} {
    \node[circle, draw] (a\i) at (1.5*\i, -2) {};
    \draw (a\i) node[above=1ex]{\footnotesize{$a_{\l}$}};
    \node[circle, draw] (d\i) at (1.5*\i, -8) {};
    \draw (d\i) node[below=1ex]{\footnotesize{$\ba_{\l}$}};
  }
  \foreach \i/\l in {0/1,1/2,2/3,7/{N-2},8/{N-1},9/N} {
    \node[circle, draw] (b\i) at (1.5*\i, -4) {};
    \draw (b\i) node[above=1ex]{\footnotesize{$b_{\l}$}};
    \node[circle, draw] (c\i) at (1.5*\i, -6) {};
    \draw (c\i) node[below=1ex]{\footnotesize{$c_{\l}$}};
  }
  \foreach \t[count=\y] in {a,b,c,d} {
    \foreach \i[count=\x] in {j,i,k} {
      \node[circle, draw] (\t\i) at (3.15+1.8*\x, -2*\y) {};
    }
    \draw ($(\t2)!.5!(\t j)$) node[anchor=mid] {${\cdots}$};
    \draw ($(\t j)!.5!(\t i)$) node[anchor=mid] {${\cdots}$};
    \draw ($(\t i)!.5!(\t k)$) node[anchor=mid] {${\cdots}$};
    \draw ($(\t k)!.5!(\t7)$) node[anchor=mid] {${\cdots}$};
  }


  \draw (ai) node[above=1ex]{\footnotesize{$a_i$}};
  \draw (bj) node[above=1ex]{\footnotesize{$b_j$}};
  \draw (ck) node[below=1ex]{\footnotesize{$c_k$}};
  \draw (di) node[below=1ex]{\footnotesize{$\ba_i$}};

  \foreach \i in {1,2,i,j,k,7,8} {
    \foreach \j in {0,1,2,i,j,k,7,8,9} {
      \path (a\i) edge[dotted,color=gray] (b\j);
      \path (d\i) edge[dotted,color=gray] (c\j);
    }
  }
  \foreach \i in {0,1,2,i,j,k,7,8,9} {
    \foreach \j in {0,1,2,i,j,k,7,8,9} {
      \path (b\i) edge[dotted,color=gray] (c\j);
    }
  }

  \draw (a1) -- node[above] {$0$} (a2);
  \draw (a7) -- node[above] {$0$} (a8);
  \draw (d1) -- node[below] {$0$} (d2);
  \draw (d7) -- node[below] {$0$} (d8);
  \foreach \t in {a,d}{
    \foreach \x/\y in {2/j,j/i,i/k,k/7}{
      \draw (\t\x) -- ($(\t\x)!.33!(\t\y)$);
      \draw (\t\y) -- ($(\t\y)!.33!(\t\x)$);
    }
  }


  \draw (ai) -- node[right] {$w(a_i,b_j) + 3(P+1- i)W$} (bj);
  \draw (bj) -- node[right] {\ $w(b_j,c_k) + 3W$} (ck);
  \draw (ck) -- node[right] {$w(c_k,a_i) + 3iW$} (di);

  \node [circle,fill,text=white,inner sep=2pt] at (a1) {$\mathbf{s}$};
  \node [circle,fill,text=white,inner sep=2pt] at (d8) {$\mathbf{t}$};
\end{tikzpicture}
\caption{The graph created in the reduction from Negative Triangle.}\label{fig:negativetriangle}
\end{figure}
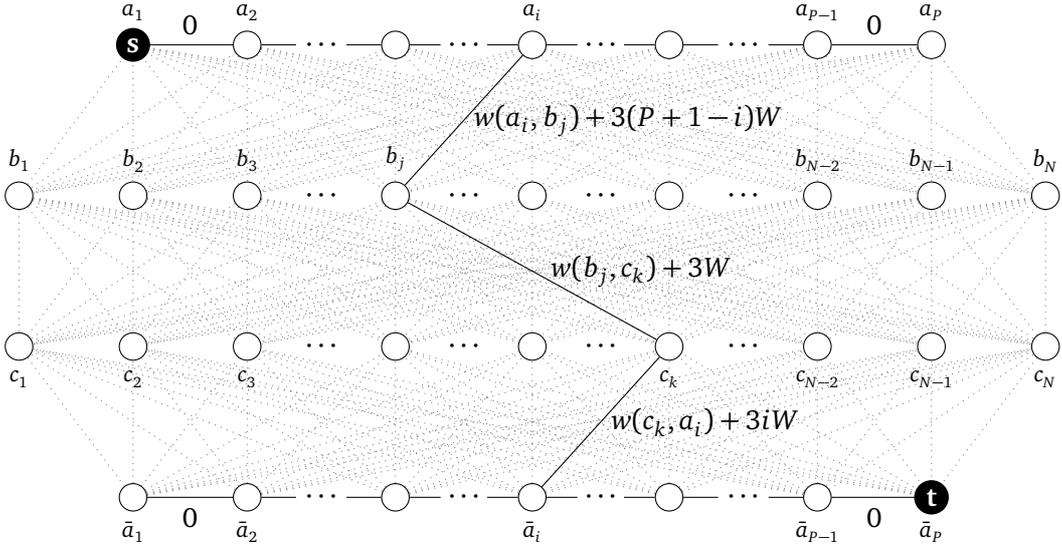

Let $P=|A|$. Suppose that $A=\{a_1,\ldots,a_P\}$, $B=\{b_1,\ldots,b_N\}$, and $C=\{c_1,\ldots, c_N\}$. Moreover, let $W$ denote the maximum absolute value of an edge weight. Create an undirected graph (see Figure~\ref{fig:negativetriangle}) with node set $A\cup B \cup C \cup \bA$, where $\bA=\{\ba_1,\ldots,\ba_P\}$ is a disjoint copy of $A$.
\begin{itemize}[leftmargin=*,nosep]
  \item For every $a_i\in A$ and $b_j\in B$, add edge $\{a_i,b_j\}$ with weight $w(a_i,b_j)+3(P+1-i)W$.
  \item For every $b_j\in B$ and $c_k\in C$, add edge $\{b_j,c_k\}$ with weight $w(b_j,c_k)+3W$.
  \item For every $c_k\in C$ and $a_i\in A$, add edge $\{c_k,\ba_i\}$ with weight $w(c_k,a_i)+3iW$.
  \item For every $i\in \{1,\ldots,P-1\}$, add edges $\{a_i,a_{i+1}\}$ and $\{\ba_i,\ba_{i+1}\}$ with weights $0$.
\end{itemize}

Consider a shortest path in this graph from $s \defeq a_1$ to $t \defeq \ba_P$ using at most $h \defeq P+2$ hops. We claim that its total length is less than $(3P+2)W$ if and only if there is a negative triangle in the initial graph.
Indeed, each triple $(a_i,b_j,c_k)\in A\times B\times C$ corresponds to path $a_1 - a_2 - \cdots - a_i - b_j - c_k - \ba_i- \ba_{i+1} - \cdots - \ba_P$, which uses exactly $P+2$ hops and has total length 
\begin{multline*}
  w(a_i,b_j)+3(P+1-i)W + w(b_j,c_k)+3W + w(c_k,a_i)+3iW \\ = \big( w(a_i,b_j)+ w(b_j,c_k) + w(c_k,a_i) \big) +(3P+2)W.
\end{multline*}
Hence, the ``if'' direction follows. For the ``only if'' direction, fix an $s$-$t$ path with 
at most $P+2$ hops and a total length strictly less than $(3P+2)W$. The path must be of the form 
$a_1 - a_2 - \cdots - a_i - b_j - \cdots - c_k - \ba_{\bi} - \ba_{\bi+1} - \cdots - \ba_{P}$,
where $\{a_i, b_j\}$ is the first edge that leaves $A$ and $\{c_k, \ba_{\bi}\}$ is the last edge that enters $\bA$.
Every edge incident to $b_j$ or $c_k$ has weight at least $-W + 3W=2W$, and the direct edge between $b_j$ and $c_k$ has weight at most $W+3W=2\cdot 2W$. Thus, the direct edge is the cheapest walk from $b_j$ to $c_k$ both in terms of the length and the number of hops. 
Consequently, we may assume without loss of generality that our $s$-$t$ path proceeds directly from $b_j$ to $c_k$.
This means that the number of hops is $i+3+(P-1-\bi)=P+2+i-\bi$, whereas the total length is
\begin{multline*}
  w(a_i,b_j)+3(P+1-i)W + w(b_j,c_k)+3W + w(c_k,a_{\bi})+3\bi W \\ = \big(w(a_i,b_j)+ w(b_j,c_k) + w(c_k,a_{\bi})\big)+(3P+2)W+3(\bi-i)W.
\end{multline*}
If $\bi < i$, then the number of hops is at least $P+3$, which is larger than assumed.
If $\bi > i$, then the path length is at least $-3W + (3P+2)W+3W \geqslant (3P+2)W$, a contradiction.
Therefore, $i=\bi$ holds. Since the path length is less than $(3P+2)W$, we conclude that $w(a_i,b_j)+w(b_j,c_k)+w(c_k,a_i) < 0$, i.e., $(a_i,b_j,c_k)$ is a negative triangle in the initial graph.

\section{Hardness under Min-Plus Convolution Hypothesis}

\paragraph*{Preliminaries.}
In the Max-Plus Convolution Upper Bound problem, we are given three sequences $(a[i])_{i=0}^{n-1}$, $(b[i])_{i=0}^{n-1}$, and $(c[i])_{i=0}^{n-1}$ of $n$ numbers each, and the goal is to decide whether $c[k] \geqslant \max_{i+j=k} (a[i] + b[j])$ holds for all $k$. In other words, we want to find $i, j, k$ such that $i+j=k$ and $c[k] < a[i] + b[j]$. The Max-Plus Convolution Upper Bound and Min-Plus Convolution problems are equivalent under subquadratic reductions~\cite{CyganMWW19}; thus, in particular, unless the Min-Plus Convolution Hypothesis fails, there is no $O(n^{2-\varepsilon})$ time algorithm for Max-Plus Convolution Upper Bound, for any $\varepsilon > 0$.

Let us introduce an intermediate problem, which we call Common Max-Plus Convolution Upper Bound: Given $M$ pairs of sequences
\[\big((a_1[i])_{i=0}^{N-1}, (b_1[i])_{i=0}^{N-1}\big), \ldots, \big((a_M[i])_{i=0}^{N-1}, (b_M[i])_{i=0}^{N-1}\big),\]
and one sequence $(c[i])_{i=0}^{N-1}$, decide if there exist $i, j, k, \ell$ such that $i+j=k$ and $c[k] < a_\ell[i] + b_\ell[j]$. We call such $(i,j,k,\ell)$ a \emph{violating quadruple}. First, we show that the naive running time of $O(MN^2)$ is conditionally optimal for this problem.

\begin{lemma}
Unless the Min-Plus Convolution Hypothesis fails, there is no $O(MN^{2-\varepsilon})$ time algorithm for Common Max-Plus Convolution Upper Bound, for any $\varepsilon > 0$, even when restricted to instances with $M=\Theta(N^\alpha)$ for an arbitrarily chosen constant $\alpha \geqslant 0$.
\end{lemma}

\begin{proof}
The argument is based on a self-reduction of Min-Plus Convolution (see~\cite[proof of Theorem~5.5]{CyganMWW19}). Let $\beta = \nicefrac{\alpha}{(1+\alpha)} \in [0,1)$. We start with an instance of Max-Plus Convolution Upper Bound with three sequences $a$, $b$, $c$, each of length $n$. We split $a$ and $b$ into $\Theta(n^\beta)$ blocks of consecutive elements, each block of length $\Theta(n^{1-\beta})$. For every pair of blocks, one from $a$ and the other from $b$, we want to check if a corresponding fragment of $c$ is an upper bound of their max-plus convolution. Similarly to~\cite{CyganMWW19}, we add suitable padding so that all three sequences are of the same length. This way, we end up with $\Theta(n^{2\beta})$ smaller instances of Max-Plus Convolution Upper Bound. The key observation is that these instances can be grouped into $\Theta(n^\beta)$ groups, each of size $\Theta(n^\beta)$, such that all instances in a single group share the same third sequence. Each such group becomes a single instance of Common Max-Plus Convolution Upper Bound, with $M=\Theta(n^\beta)$ and $N=\Theta(n^{1-\beta})$. If each of these instances can be solved in $O(MN^{2-\varepsilon}) = O(n^{\beta+(1-\beta)(2-\varepsilon)})$ time, then the original instance can be solved in $O(n^{2\beta+(1-\beta)(2-\varepsilon)})=O(n^{2-(1-\beta)\varepsilon})$ time, and the Min-Plus Convolution Hypothesis fails. We conclude the proof by observing that $n^\beta = n^{\nicefrac{\alpha}{(1+\alpha)}} = (n^{\nicefrac{1}{(1+\alpha)}})^\alpha = (n^{1-\beta})^\alpha$, and thus $M=\Theta(N^\alpha)$ holds as desired.
\end{proof}

Cygan, Mucha, Węgrzycki, and Włodarczyk \cite[proof of Theorem~5.4]{CyganMWW19} showed that w.l.o.g.~the input sequences to Max-Plus Convolution Upper Bound are nonnegative and strictly increasing. The same argument applies to Common Max-Plus Convolution Upper Bound. Using the additional structure, we can replace the condition $i+j=k$ with $i+j \leqslant k$. Indeed, suppose we find $(i,j,k,\ell)$ with $i+j \leqslant k$ and $c[k] < a_\ell[i] + b_\ell[j]$; then, by monotonicity, $c[i+j] \leqslant c[k]$, and hence $(i, j, i+j, \ell)$ is a quadruple satisfying the original condition.

\paragraph{Reduction.}

We show how to reduce an instance of Common Max-Plus Convolution Upper Bound, with $M$ pairs of length-$N$ sequences, to finding the length of a shortest hop-bounded $s$-$t$ path in an undirected graph with $n = \Theta(N+M)$ nodes, $m = \Theta(NM)$ edges, and the hop bound $h = \Theta(N)$. This will let us conclude that an $O(h^{1-\varepsilon}m)$ time (or $O(hm^{1-\varepsilon})$ time) shortest path algorithm would give an $O(MN^{2-\varepsilon})$ time (respectively, $O(M^{1-\varepsilon}N^{2-\varepsilon})$ time) algorithm for the Common Max-Plus Convolution Upper Bound problem and thus violate the Min-Plus Convolution Hypothesis.

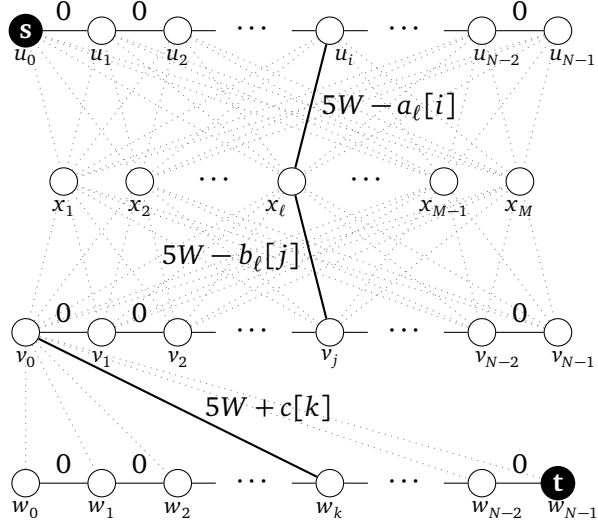
\begin{figure}
\centering
\begin{tikzpicture} 

  \foreach \i in {0,...,2} {
    \node [font=\footnotesize] at (\i,-2.35) {$u_\i$};
    \node [font=\footnotesize] at (\i,-6.35) {$v_\i$};
    \node [font=\footnotesize] at (\i,-8.35) {$w_\i$};
  }

  \node [font=\footnotesize] at (4.2,-2.35) {$u_i$};
  \node [font=\footnotesize] at (4  ,-6.35) {$v_j$};
  \node [font=\footnotesize] at (4  ,-8.35) {$w_k$};

  \node [font=\footnotesize] at (6.2,-2.35) {$u_{N-2}$};
  \node [font=\footnotesize] at (7.2,-2.35) {$u_{N-1}$};
  \node [font=\footnotesize] at (6.2,-6.35) {$v_{N-2}$};
  \node [font=\footnotesize] at (7.2,-6.35) {$v_{N-1}$};
  \node [font=\footnotesize] at (6.2,-8.35) {$w_{N-2}$};
  \node [font=\footnotesize] at (7.2,-8.35) {$w_{N-1}$};

  \node [font=\footnotesize] at (0.5,-4.35) {$x_1$};
  \node [font=\footnotesize] at (1.5,-4.35) {$x_2$};
  \node [font=\footnotesize] at (3.3,-4.35) {$x_\ell$};
  \node [font=\footnotesize] at (5.5,-4.35) {$x_{M-1}$};
  \node [font=\footnotesize] at (6.5,-4.35) {$x_{M}$};
  
  \foreach \l in {1,3,4} {
    
    \foreach \i in {0,...,2} {
        \node [circle,draw] (v\i\l) at (\i,-2*\l) {};
    }
    \node at (3, -2*\l) {$\cdots$};
    \node [circle,draw] (vi\l) at (4,-2*\l) {};
    \node at (5, -2*\l) {$\cdots$};
    \node [circle,draw] (vm\l) at (6,-2*\l) {};
    \node [circle,draw] (vn\l) at (7,-2*\l) {};

    \path (v2\l) edge (2.5,-2*\l);
    \path (vi\l) edge (3.5,-2*\l);
    \path (vi\l) edge (4.5,-2*\l);
    \path (vm\l) edge (5.5,-2*\l);
  }

  \node [circle,draw] (x1) at (0.5,-4) {};
  \node [circle,draw] (x2) at (1.5,-4) {};
  \node                    at (2.5,-4) {$\cdots$};
  \node [circle,draw] (xi) at (3.5,-4) {};
  \node                    at (4.5,-4) {$\cdots$};
  \node [circle,draw] (xm) at (5.5,-4) {};
  \node [circle,draw] (xn) at (6.5,-4) {};

  \foreach \fromA in {v0,v1,v2,vi,vm,vn} {
    \foreach \fromB in {1,3} {
      \foreach \to in {x1,x2,xi,xm,xn} {
        \path (\fromA\fromB) edge[dotted, color=gray] (\to);
      }
    }
  }

  \foreach \to in {v04,v14,v24,vi4,vm4,vn4} {
    \path (v03) edge[dotted, color=gray] (\to);
  }

  \foreach \l in {1,3,4} {

    \path (v0\l) edge node[above] {$0$} (v1\l);
    \path (v1\l) edge node[above] {$0$} (v2\l);
    \path (vm\l) edge node[above] {$0$} (vn\l);
  }
  

  \path (vi1) edge[thick] node[right] {$5W-a_\ell[i]$} (xi);
  \path (vi3) edge[thick] node[left] {$5W-b_\ell[j]$} (xi);
  \path (v03) edge[thick] node[right] {\ \ $5W+c[k]$} (vi4);

  \node [circle,fill,text=white,inner sep=2pt] at (0,-2) {$\mathbf{s}$};
  \node [circle,fill,text=white,inner sep=2pt] at (7,-8) {$\mathbf{t}$};

\end{tikzpicture}
\caption{The graph created in the reduction from Common Max-Plus Convolution Upper Bound.}\label{fig:maxplus}
\end{figure}

Let $W$ denote the maximum value of any input sequence element. Create an undirected graph (see Figure~\ref{fig:maxplus}) composed of three paths $u_0 - u_1 - \cdots - u_{N-1}$, $v_0 - v_1 - \cdots - v_{N-1}$, $w_0 - w_1 - \cdots - w_{N-1}$, and $M$ isolated nodes $x_1,x_2,\ldots,x_M$. Set the weights of all the path edges to $0$.
For every $i \in \{0, 1, \ldots, N-1\}$ and $\ell \in \{1, 2, \ldots, M\}$, add an edge between $u_i$ and $x_\ell$ with weight $5W - a_\ell[i]$. Then, for every $j \in \{0, 1, \ldots, N-1\}$ and $\ell \in \{1, 2, \ldots, M\}$, add an edge between $x_\ell$ and $v_j$ with weight $5W - b_\ell[j]$. Finally, for every $k \in \{0, 1, \ldots, N-1\}$, add an edge between $v_0$ and $w_k$ with weight $5W + c[k]$.

Consider a shortest path in this graph from $s \defeq u_0$ to $t \defeq w_{N-1}$ using at most $h \defeq N+2$ hops. We claim that its total length is less than $15W$ if and only if there is a quadruple~$(i, j, k, \ell)$ with $i+j \leqslant k$ and $c[k] < a_\ell[i] + b_\ell[j]$.

Indeed, each quadruple $(i,j,k,\ell)$ corresponds to path
\begin{equation}
u_0 - u_1 - \cdots - u_i - x_\ell - v_j - v_{j-1} - \cdots - v_0 - w_k - w_{k+1} - \cdots - w_{N-1}.
\tag{$\star$}
\label{path}
\end{equation}
Such a path uses $i + 1  + 1 + j + 1 + (N -1- k) = (N+2) + (i + j - k)$ hops and has total length $15W - a_\ell[i] - b_\ell[j] + c[k]$. Hence, the ``if'' direction follows. For the ``only if'' direction, since every non-zero edge weight is at least $4W$, an $s$–$t$ path of total length less than $15W$ can use at most three such edges, and therefore it must be of the form~(\ref{path}). The hop bound implies $i+j-k \leqslant 0$ and the total length bound implies $c[k] - a_\ell[i] - b_\ell[j] < 0$.

Recall that $n = \Theta(N+M)$, $m = \Theta(NM)$, and $h = \Theta(N)$. For $\eta \in [\nicefrac{1}{2}, 1]$, in order to get hard shortest hop-bounded path instances with density $\Theta(m^\eta)$ and hop bound $h=\Theta(m^\eta)$, we start with the Common Max-Plus Convolution Upper Bound problem restricted to instances with $M = \Theta(N^{\nicefrac{(1-\eta)}{\eta}})$. This implies $h = \Theta(N) = \Theta((NM)^\eta) = \Theta(m ^ \eta)$. Moreover, due to $\eta \geqslant \nicefrac{1}{2}$, we have $M \leqslant O(N)$, and thus $n = \Theta(N + M) = \Theta(N) = \Theta(m^\eta)$ holds as desired. This concludes the proof of \cref{thm:minplus}.

\section*{Acknowledgements}

Adam Polak would like to thank Danupon Nanongkai and Luca Trevisan for bringing his attention to the problem discussed in this paper, Alexandra Lassota for useful feedback on an early draft of the manuscript, and anonymous reviewers for interesting suggestions on extending the results. Part of this work was done at École Polytechnique Fédérale de Lausanne, supported by the Swiss National Science Foundation projects \emph{Lattice Algorithms and Integer Programming} (185030) and \emph{Complexity of Integer Programming} (CRFS-2\_207365).

\bibliographystyle{alphaurl}
\bibliography{main}

\end{document}